\documentclass[runningheads]{llncs}
\usepackage{enumitem}

\usepackage{graphicx,latexsym,amsfonts,amssymb,amsmath,subfigure}
\usepackage{xcolor}
\usepackage{fmtcount}
\usepackage{thm-restate}
\usepackage{hyperref}

\spnewtheorem{observation}{Observation}{\bfseries}{\itshape}

\usepackage{dsfont}
\newcommand{\N}{\ensuremath{\mathds{N}}} %
\renewcommand{\orcidID}[1]{\hypersetup{hidelinks}\hspace*{1pt}\href{https://orcid.org/#1}{\includegraphics[width=8pt]{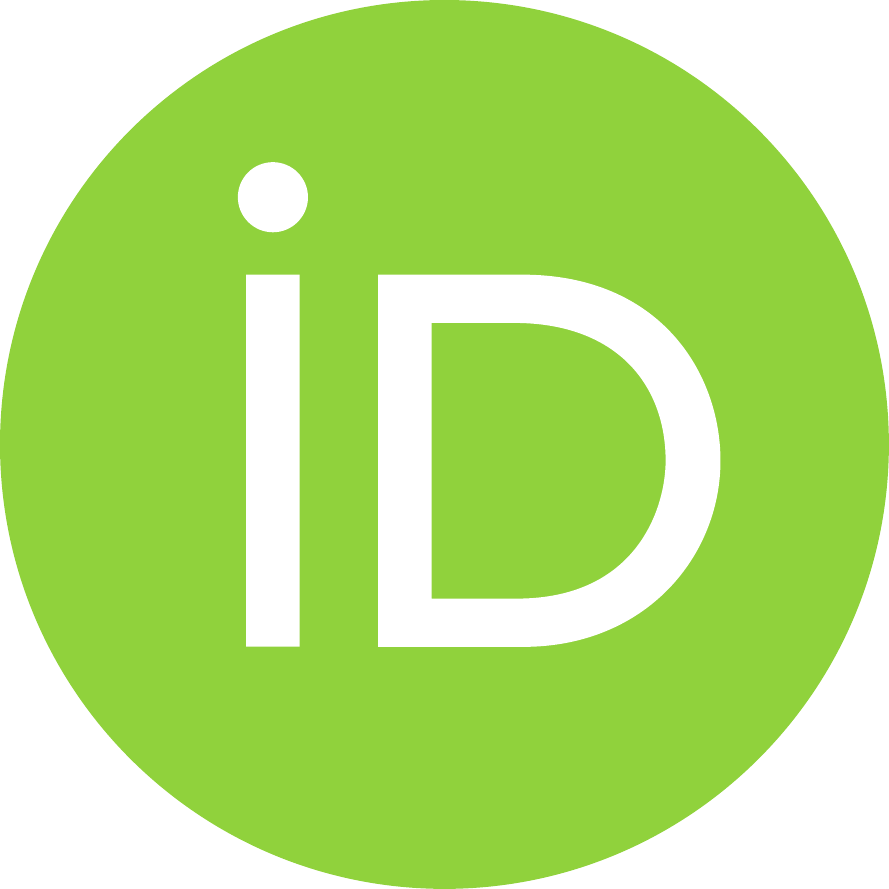}}}

\title{Simple Topological Drawings of\\
  $k$-Planar Graphs\thanks{This work was initiated at the
    $17^{th}$ Gremo Workshop on Open Problems (GWOP) 2019. The authors thank
    the organizers of the workshop for inviting us and providing a
    productive working atmosphere.
    M.~H.\ and M.~M.~R.\ are supported by the Swiss National Science Foundation within the collaborative DACH project \emph{Arrangements and Drawings} as SNSF Project 200021E-171681.
    Research by C.~D.~T.\ was supported in part by the NSF award DMS-1800734.}}
\author{Michael~Hoffmann\inst{1}\orcidID{0000-0001-5307-7106}
\and Chih-Hung~Liu\inst{1}\orcidID{0000-0001-9683-5982}
\and Meghana~M.~Reddy\inst{1}\orcidID{0000-0001-9185-1246}
\and Csaba~D.~T\'oth\inst{2,3}\orcidID{0000-0002-8769-3190}
}

\institute{%
  Department of Computer Science, ETHZ, Z\"urich, Switzerland\\
  \email{\{hoffmann,chih-hung.liu,meghana.mreddy\}@inf.ethz.ch}
  \and
  Department of Mathematics, Cal State Northridge,
  Los Angeles, CA, USA
  \and Department of Computer Science, Tufts
  University, Medford, MA, USA\\
  \email{csaba.toth@csun.edu}
}

\begin{document}

\maketitle

\begin{abstract}
  Every finite graph admits a \emph{simple (topological) drawing},
  that is, a drawing where every pair of edges intersects in at most
  one point.
  However, in combination with other restrictions simple drawings do
  not universally exist. For instance, \emph{$k$-planar graphs} are
  those graphs that can be drawn so that every edge has at most $k$
  crossings (i.e., they admit a \emph{$k$-plane drawing}).
  It is known that for $k\le 3$, every $k$-planar graph
  admits a $k$-plane simple drawing. But for $k\ge 4$, there
  exist $k$-planar graphs that do not admit a $k$-plane simple
  drawing. Answering a question by Schaefer, we show that there exists
  a function $f:\N\rightarrow\N$ such that every $k$-planar graph
  admits an $f(k)$-plane simple drawing, for all $k\in\N$. Note that
  the function $f$ depends on $k$ only and is independent of the size
  of the graph. Furthermore, we develop an algorithm to show
  that every $4$-planar graph admits an $8$-plane simple drawing.

\end{abstract}

\section{Introduction}

A \emph{topological drawing} of a graph $G$ in the plane is a
representation of $G$ in which the vertices are mapped to pairwise distinct points in the plane and edges are mapped to Jordan arcs
that do not pass through (the images of) vertices. Moreover, no three Jordan arcs pass through the same point in the plane, and every pair of Jordan arcs has finitely many intersection points, each of which is either a common endpoint or a \emph{crossing}, where the two arcs cross transversally.
A graph is \emph{$k$-planar} if it admits a topological drawing in the
plane where every edge is crossed at most $k$ times, and such a drawing is called a \emph{$k$-plane drawing}. A \emph{simple topological drawing} of a graph refers to a topological drawing where no two edges cross more than once and no two adjacent edges cross. We study simple topological drawings of $k$-planar graphs.

It is well known that drawings of a graph $G$ that attain the minimum number of crossings (i.e., the \emph{crossing number} 
of $G$) are simple topological drawings~\cite[p.~18]{schaefer2012graph}. However, a drawing that minimizes the total number of crossings need not minimize the maximum number of crossings per edge; and a drawing that minimizes the maximum number of crossings per edge need not be simple. A \emph{$k$-plane simple topological drawing} is a simple topological drawing where every edge is crossed at most $k$ times. We study the simple topological drawings of $k$-planar graphs and prove that there exists a function $f:\N\rightarrow\N$ such that every $k$-planar graph admits an $f(k)$-plane simple topological drawing by designing an algorithm to obtain the simple topological drawing from a $k$-plane drawing. 
The function $f$ in our bound is exponential in $k$, more precisely $f(k)\in O^*(3^k)$.
It remains open whether this can be improved to a bound that is polynomial in $k$.
We also present a significantly better bound for $4$-planar graphs.

In a $k$-plane drawing 
adjacent edges may cross, and two 
edges may
cross many times. To obtain a simple topological drawing, we need
to eliminate crossings between adjacent edges and ensure that any two
edges cross at most once.

\paragraph{Related Work.}
It is easy to see that every $1$-planar graph admits a $1$-plane simple topological drawing~\cite{Ringel65}.
Pach et al.~\cite[Lemma~1.1]{PachRTT06} proved that every $k$-planar
graph for $k \leq 3$ admits a $k$-plane simple topological drawing. 
However, these results do not extend to $k$-planar graphs, for $k > 3$. In fact, Schaefer~\cite[p.~57]{schaefer2012graph} constructed $k$-planar graphs that do not admit a $k$-plane simple topological drawing for $k=4$. The construction idea can be extended to all $k > 4$. The \emph{local crossing number} $\text{lcr}(G)$ of a graph $G$ is the minimum integer $k$ such that $G$ admits a drawing where every edge has at most $k$ crossings. The \emph{simple} local crossing number $\text{lcr}^*(G)$ minimizes $k$ over all simple topological drawings of $G$. Schaefer~\cite[p.~59]{schaefer2012graph} asked whether the $\text{lcr}^*(G)$ can be bounded by a function of $\text{lcr}(G)$. We answer this question in the affirmative and show that there exists a function $f:\N\rightarrow\N$ such that $\text{lcr}^*(G)\leq f(\text{lcr}(G))$.

The family of $k$-planar graphs, for small values of $k$, was instrumental in proving the current best bounds on the multiplicative constant in the Crossing Lemma and the Szemer\'edi-Trotter theorem on point-line incidences~\cite{ackerman2019topological,PachRTT06}. Ackerman~\cite{ackerman2019topological} showed that every  graph with $n\geq 3$ vertices that admits a simple 4-plane drawing has at most $m\leq 6n-12$ edges, and claims that this bound holds for all 4-planar graphs. Pach et al.~\cite[Conjecture~5.4]{PachRTT06} conjectured that for all $k,n\geq 1$, the maximum number of edges in a $k$-planar $n$-vertex graph is attained by a graph that admits a simple $k$-plane drawing.

\section{Preliminaries}

\paragraph{Lenses in topological drawings.}
We start with definitions needed to describe the key operations in our algorithms. In a topological drawing, we define a structure called \textit{lens}. Consider two edges, $e$ and $f$, that intersect in two distinct points, $\alpha$ and $\beta$ (each of which is either a common endpoint or a crossing). Let $e_{\alpha\beta}$ (resp., $f_{\alpha\beta}$) denote the portion of $e$ (resp., $f$) between $\alpha$ and $\beta$. The arcs $e_{\alpha\beta}$ and $f_{\alpha\beta}$ together are called a \textit{lens} if $e_{\alpha\beta}$ and $f_{\alpha\beta}$ do not intersect except at $\alpha$ and $\beta$. See \figurename~\ref{fig:fig1} for examples. The lens is denoted by $L(e_{\alpha\beta},f_{\alpha\beta})$. A lens $L(e_{\alpha\beta},f_{\alpha\beta})$ is bounded by \emph{independent arcs} if both $\alpha$ and $\beta$ are crossings, else (if $\alpha$ or $\beta$ is a vertex of $G$) it is bounded by \emph{adjacent arcs}.

\begin{restatable}{lemma}{lemmaone}\label{lemma:existence_of_lens}
If a pair of edges $e$ and $f$ intersect in more than one point, then there exist arcs
$e_{\alpha\beta}\subset e$ and $f_{\alpha\beta}\subset f$ that form a lens.
\end{restatable}

\begin{figure}[ht]
  \centering
  \subfigure[]{
	 \label{fig:fig1_a}
	\includegraphics[scale=1]
	{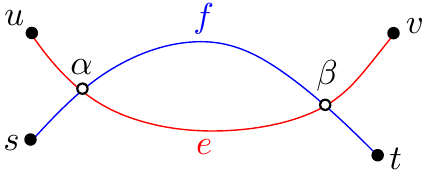}}
  \qquad\qquad
  \subfigure[]{
	 \label{fig:fig1_b}
	\includegraphics[scale=1]
	{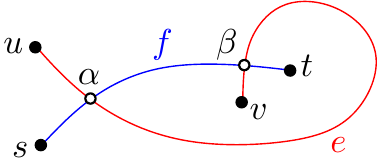}}

    \caption{Lenses formed by two edges.}
  \label{fig:fig1}
\end{figure}


\paragraph{Operations.}
We present 
algorithms that transform a $k$-plane drawing into a simple
topological drawing by a sequence of elementary operations. Each
operation modifies one or two edges that form a lens so that the lens
is eliminated. We use two elementary operations, \textsc{Swap} and
\textsc{Reroute}. Both have been used previously (e.g., in
\cite[Lemma~1.1]{PachRTT06}); we describe them here for completeness.

The common setup in both operations is the following. Let $e=uv$ and
$f=st$ be edges that form a lens $L(e_{\alpha\beta},f_{\alpha\beta})$,
where $\alpha$ and $\beta$ are each a crossing or a common endpoint.
Assume that the Jordan arc of $e$ visits $u$, $\alpha$,
$\beta$, $v$, and the Jordan arc of $f$ visits $s$, $\alpha$, $\beta$,
and $t$ in this order. Let $\overline{\alpha}$ and $\overline{\beta}$
be sufficiently small disks centered at $\alpha$ and $\beta$, resp.,
so that their boundary circles each intersect $e$ and $f$ twice, but
do not intersect any other edge.

\paragraph{Swap operation.}
We modify the drawing of $e$ and $f$ in three steps as follows.
(1)~Redraw $e$ such that it follows its current arc from $u$ to $\alpha$, then continues along $f_{\alpha\beta}$ to $\beta$, and further to $v$ along its original arc. Similarly, redraw $f$ such that it follows its current arc from $s$ to $\alpha$, then continues along $e_{\alpha\beta}$ to $\beta$, and further to $t$ along its original arc.
(2)~Replace the portion of $e$ and $f$ in $\overline{\alpha}$ and $\overline{\beta}$ by straight line segments.
(3)~Eliminate self-crossings, if any is introduced, by removing any loops from the modified arcs of $e$ and $f$.
The swap operation is denoted by \textsc{Swap}$(e_{\alpha\beta},f_{\alpha\beta})$; see \figurename~\ref{fig:lens2} for illustrations.
The swap operation for a lens bounded by adjacent arcs is defined similarly.

\begin{figure}[ht]
  \centering
  \subfigure[]{
	 \label{fig:lens1_img2}
	\includegraphics[scale=1]
	{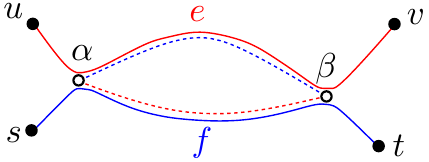}}
  \qquad\qquad
  \subfigure[]{
	 \label{fig:lens2_img2}
	\includegraphics[scale=1]
	{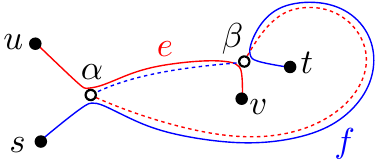}}

    \caption{\textsc{Swap}$(e_{\alpha\beta},f_{\alpha\beta})$ applied to the two lenses in \figurename~\ref{fig:fig1}}
  \label{fig:lens2}
\end{figure}

%

\begin{restatable}{observation}{obsone}\label{obs:exchange}
Let $D$ be a topological drawing of a graph $G$, and let $L(e_{\alpha\beta},f_{\alpha\beta})$ be a lens. Operation \textsc{Swap}$(e_{\alpha\beta},f_{\alpha\beta})$ produces a topological drawing that has at least one fewer crossing than $D$.
\end{restatable}

\paragraph{Reroute operation.}
We modify the drawing of $f$ in three steps as follows.
(1)~Redraw $f$ such that it follows its current arc from $s$ to the first intersection with $\overline{\alpha}$, it does not cross $e$ in $\overline{\alpha}$, and then it closely follows arc $e_{\alpha\beta}$ to $\overline{\beta}$, and further follows its original arc from $\overline{\beta}$ to $t$.
(2)~Replace the portion of $f$ in the interior of $\overline{\alpha}$ and $\overline{\beta}$ by straight line segments.
(3)~Eliminate self-crossings, if any are introduced, by removing any loops from the modified arc of $f$.
The reroute operation is denoted by \textsc{Reroute}$(e_{\alpha\beta},f_{\alpha\beta})$; see \figurename~\ref{fig:lens3} for illustrations.
The reroute operation for a lens bounded by adjacent arcs is defined similarly.

\begin{figure}[ht]
  \centering
  \subfigure[]{
	 \label{fig:lens1_img3}
	\includegraphics[scale=1]
	{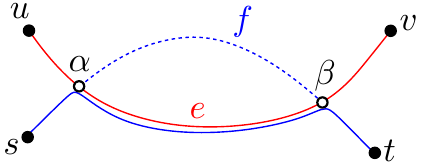}}
  \qquad
  \subfigure[]{
	 \label{fig:lens2_img3}
	\includegraphics[scale=1]
	{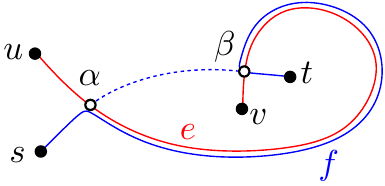}}

    \caption{\textsc{Reroute}$(e_{\alpha\beta},f_{\alpha\beta})$ operation on the two lenses in \figurename~\ref{fig:fig1}}
  \label{fig:lens3}
\end{figure}

\begin{restatable}{observation}{obstwo}\label{obs:reroute}
Let $D$ be a topological drawing of a graph $G$, and let $L(e_{\alpha\beta},f_{\alpha\beta})$ be a lens. Operation \textsc{Reroute}$(e_{\alpha\beta},f_{\alpha\beta})$ produces a topological drawing.
\end{restatable}
While a \textsc{Reroute}$(e_{\alpha\beta},f_{\alpha\beta})$
operation modifies only the edge $f$, 
it may increase the
total number of crossings, as well as the number of crossings on 
$f$.

\paragraph{Planarization.}
Let $D$ be a topological drawing of a graph $G$. Denote by $N$ the
planarization of $D$ (i.e., we introduce a vertex of degree four at
every crossing in $D$). We call this graph a \textit{network}. We
refer to the vertices and edges of $N$ as \emph{nodes} and
\emph{segments}, respectively, so as to distinguish them from the
corresponding entities in $G$. Our algorithms in
Section~\ref{sec:general}--\ref{sec:8bound} use the planarization $N$
of a drawing $D$, then successively modify the drawing $D$, and
ultimately return a simple topological drawing of $G$. We formulate
invariants for these algorithms in terms of the planarization $N$ of
the initial drawing. In other words, $N$ remains fixed (in particular,
$N$ will not be the planarization of the modified drawings). As
\textsc{Reroute} operations redraw edges to closely follow existing
edges, our algorithms will maintain the following invariants:

\begin{enumerate}[label={(I\arabic*)}]\setlength{\itemindent}{\labelsep}
\item\label{I1} Every edge in $D$ closely follows a path in the network $N$;
\item\label{I2} every pair of edges in $D$ cross only in a small neighborhood of a node of $N$;
\item\label{I3} every pair of edges crosses at most once in each such neighborhood.
\end{enumerate}

\paragraph{Length of an arc and number of crossings.} Let $a$ be a Jordan arc that closely follows a path in $N$ such that its endpoints are in the small neighborhoods of nodes of $N$. The \emph{length} of $a$, denoted by $\ell(a)$, is the graph-theoretic length of the path of $N$ that $a$ closely follows. Let $x(a)$ denote the number of crossings on the arc $a$ in a drawing $D$. Note that the length $\ell(a)$ is measured in terms of the (fixed) network $N$, and $x(a)$ is measured in terms of the (varying) drawing $D$. For instance, in \figurename~\ref{fig:lens2_img3} we have $\ell(f)=3$ both before and after rerouting, whereas $x(f)=2$ before and $x(f)=1$ after rerouting.


\section{General Bound for $k$-Planar Graphs}
\label{sec:general}

In this section we describe and analyze an algorithm to transform a
topological drawing into a simple topological drawing whose local
crossing number is bounded by a function of the local crossing number
of the original drawing.

\begin{description}
\item[Algorithm~1.]

\item[] Let $D_0$ be a topological $k$-plane
  drawing of a graph $G=(V,E)$. Let $N$ be the planarization of $D_0$.
  Let $D:=D_0$. 

\item[] While there exists a lens in $D$, do the following.
  \begin{description}
  \item Let $L(e_{\alpha\beta},f_{\alpha\beta})$ be a lens so that
    w.l.o.g.\ $\ell(e_{\alpha\beta})<\ell(f_{\alpha\beta})$, or
    $\ell(e_{\alpha\beta})=\ell(f_{\alpha\beta})$ and
    $x(e_{\alpha\beta})\leq x(f_{\alpha\beta})$. Modify $D$ by
    applying \textsc{Reroute}$(e_{\alpha\beta},f_{\alpha\beta})$.
\end{description}

\item[] When the while loop terminates, return the drawing $D$.
\end{description}

\begin{restatable}{observation}{obsthree}\label{obs:length}
Algorithm~1 maintains invariants \ref{I1}--\ref{I3}, and the length of every edge decreases or remains the same.
\end{restatable}

\begin{corollary}\label{cor:length}
Algorithm~1 maintains the following invariant:
\begin{enumerate}[resume*]\setlength{\itemindent}{\labelsep}
\item\label{I4} The length of every edge in $D$ is at most $k+1$.
\end{enumerate}
\end{corollary}

\begin{lemma}\label{lem:termination}
Algorithm~1 terminates and transforms a $k$-plane topological drawing into a simple topological drawing of $G$.
\end{lemma}
\begin{proof}
Let the sum of lengths of all edges in the drawing be defined as the \emph{total length} of the drawing (recall that the length of an edge is the length of the corresponding path in $N$).
By Observation~\ref{obs:length}, the total length of the drawing monotonically decreases.
If the total length remains the same in one iteration of the while loop,
then $\ell(e_{\alpha\beta})=\ell(f_{\alpha\beta})$ and $x(e_{\alpha\beta})\leq x(f_{\alpha\beta})$. Since \textsc{Reroute}$(e_{\alpha\beta},f_{\alpha\beta})$ eliminates a crossing at $\alpha$ or $\beta$, the total number of crossings strictly decreases in this case.
Thus, the algorithm terminates. By Observations~\ref{obs:exchange}--\ref{obs:reroute}, the algorithm maintains a topological drawing. The drawing $D'$ returned by the algorithm does not contain lenses.
By Lemma~\ref{lemma:existence_of_lens}, any two edges in $D'$ intersect in at most one point. Consequently, $D'$ is a simple topological drawing of $G$.
\end{proof}


\begin{lemma}[Crossing Lemma {\cite[Theorem~6]{ackerman2019topological}}] \label{thm:crossing_lemma}
Let $G$ be a graph with $n$ vertices and $m$ edges and $D$ be a topological drawing of $G$. Let $\text{cr}(D)$ be defined as the total number of crossings in $D$, and $\text{cr}(G)$ be defined as the minimum of $\text{cr}(D)$ over all drawings $D$ of $G$. If $m \geq 6.95n$, then $\text{cr}(G) \geq \frac{1}{29} \frac{m^3}{n^2}$.
\end{lemma}

\begin{theorem}\label{thm:k-planar}
  There exists a function $f(k)$ such that every $k$-planar graph
  admits an $f(k)$-plane simple topological drawing, and there exists
  an algorithm to obtain an $f(k)$-plane simple topological drawing
  from a given $k$-plane drawing of a graph.
\end{theorem}
\begin{proof}
  The statement holds for $k\le 3$ with
  $f(k)=k$~\cite[Lemma~1.1]{PachRTT06}. Hence we may suppose that
  $k\ge 4$. Consider the drawing $D'$ returned by Algorithm~1, and a
  node $\gamma$ of the network $N$ that corresponds to a crossing. We
  analyse the subgraph $G_\gamma$ of $G$ formed by the edges of $G$
  that in $D'$ pass through a small neighborhood $\overline{\gamma}$
  of $\gamma$. Let $n_\gamma$ and $m_\gamma$ be the number of vertices
  and edges of $G_\gamma$, respectively.  By \ref{I4}, every edge in
  $D'$ corresponds to a path of length at most $k+1$ in $N$. If an
  edge $uv$ passes through $\overline{\gamma}$ in $D'$, then $N$
  contains a path of length at most $k$ from $\gamma$ to $u$ (resp.,
  $v$) in which internal vertices correspond to crossings in
  $D_0$. Every node in $N$ that corresponds to a crossing has degree
  4. Hence 
  the number of vertices reachable from $\gamma$
  on such a path is 
  $n_\gamma \leq 4\cdot 3^{k-1}$.

  We apply Lemma~\ref{thm:crossing_lemma} to the graph $G_\gamma$, and distinguish between
  two cases:
Either $m_\gamma < 6.95n$, otherwise
$m_\gamma \geq 6.95n$ and then $\text{cr}(G_\gamma) \geq \frac{1}{29} m_\gamma^3/n_\gamma^2$. Since $G_\gamma$ has $m_\gamma$ edges and each edge has at most $k$ crossings in $D$, we obtain $\frac{1}{29} m_\gamma^3/n_\gamma^2 \leq m_\gamma k/2$, which implies $m_\gamma \leq \sqrt{29k/2} \, n_\gamma$. The combination of both cases yields an upper bound $m_\gamma \leq \max\{6.95n_\gamma, \sqrt{29k/2} \, n_\gamma\}$. 
So, for $k \geq 4$ we have $m_\gamma \leq \sqrt{29k/2} \, n_\gamma$.
%

Since $m_\gamma$ edges pass through $\overline{\gamma}$, by invariant~\ref{I3} every edge passing through $\overline{\gamma}$ has at most  $m_\gamma-1$ crossings at $\overline{\gamma}$. By invariant~\ref{I4}, every edge in $G$ passes through (the neighborhood of) at most $k$ nodes of $N$. By \ref{I2}, an edge passing through $\overline{\gamma_1},\ldots ,\overline{\gamma_k}$ crosses at most $\sum_{i=1}^k(m_{\gamma_i}-1)$ edges in $D'$. Combining the upper bounds on $m_\gamma$ and $n_\gamma$, we obtain that every edge in the output drawing $D'$ has at most
$\sqrt{29k/2} \cdot 4k \cdot 3^{k-1}=\frac23 \sqrt{58} \cdot k^{3/2} \cdot 3^k$ crossings, for $k \geq 4$.
 \end{proof}

\section{An Upper Bound for 4-Planar Graphs}
\label{sec:8bound}

The function $f$ from our proof of Theorem~\ref{thm:k-planar} yields
\[
  f(4)=\frac23 \sqrt{58} \cdot 4^{3/2} \cdot 3^4\approx 3290.01
\]
and so every $4$-plane drawing can be transformed into a
$3290$-plane simple topological drawing. In this section we improve
this upper bound and show that $8$ crossings per edge suffice.

\begin{theorem} \label{theorem:4planar_8bound}
Every $4$-planar graph admits an $8$-plane simple topological drawing. Given a $4$-plane drawing of a graph with $n$ vertices, an $8$-plane simple topological drawing can be computed in $O(n)$ time.
\end{theorem}

The proof of Theorem~\ref{theorem:4planar_8bound} is constructive: Let $D_0$ be a 4-plane drawing of a 4-planar graph $G=(V,E)$ with $n=|V|$ vertices and $m=|E|$ edges. The Crossing Lemma implies that a $k$-planar graph on $n$ vertices has at most $3.81\sqrt{k}n$ edges. For $k=4$, this implies $m\leq 7.62n$. (We note that
Ackerman~\cite{ackerman2019topological} proved a bound $m\leq 6n-12$ for $4$-plane simple topological drawings with $n\geq 3$ vertices; this bound is not applicable here.)

We want to 
eliminate all lenses 
using swap and reroute operations. We define three types of special lenses that we handle separately. A lens $L(e_{\alpha\beta},f_{\alpha\beta})$ is
\begin{itemize}
\item a \emph{0-lens} if $e_{\alpha\beta}$ has no crossings;
\item a \emph{quasi-0-lens} if the arc $e_{\alpha\beta}$ has exactly one
    crossing $\gamma$, where $e$ crosses an edge $h$,
    the edges $h$ and $f$ have a common endpoint $s$,
    and the arcs $f_{s\alpha}$ and $h_{s\alpha}$ cross the same edges in the same order
    (see \figurename~\ref{fig:quasi0lens-1} for an example);
\item a \emph{1-3-lens} if  $x(e)=4$, $x(e_{\alpha\beta})=1$, and $x(f_{\alpha\beta})=3$; see \figurename~\ref{fig:lemma_for_4_img1}.
\end{itemize}
We show that all lenses other than 
0-lenses and 1-3-lenses can be eliminated by swap operations while
maintaining a 4-plane drawing
(Lemma~\ref{lem:exchange}). 
And 0-lenses can easily be eliminated by reroute operations
(Lemma~\ref{lem:0lens}). The same holds for quasi-0-lenses
(Lemma~\ref{lem:q0lens}), which are of no particular concern in the
initial drawing but are important for the analysis of the last phase
of our algorithm. The main challenge is to eliminate 1-3-lenses, which
we do by rerouting the arc with 3 crossings along the arc with 1
crossing.

Our algorithm proceeds in three phases: Phase~1 
eliminates all lenses other than 1-3-lenses. We show that it maintains
a 4-plane drawing (Lemma~\ref{lem:pre}). Phase~2 eliminates every
1-3-lens using reroute operations. We show that this phase produces an
8-plane drawing. Phase~2 may also create new lenses, but only 0- and
quasi-0-lenses, which are eliminated in Phase~3 without creating any
new lenses.

\begin{figure}[bht]
  \centering
  \subfigure{
	 \label{fig:lemma_for_4_img1}
	\includegraphics[scale=1]
	{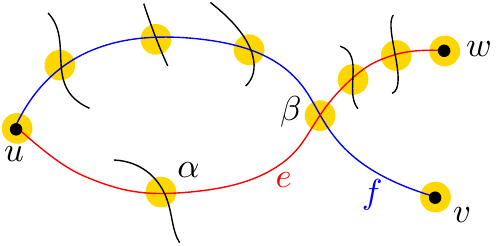}}
  \hfill
  \subfigure{
	 \label{fig:lemma_for_4_img2}
	\includegraphics[scale=1]
	{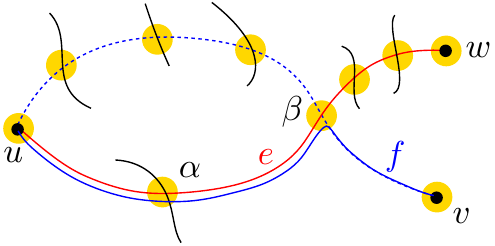}}

    \caption{\textsc{Reroute}$(e_{u\beta},f_{u\beta})$ applied to a
      1-3-lens $L(e_{u\beta},f_{u\beta})$.}
  \label{fig:lemma_for_4_figure}
\end{figure}

The initial 4-plane drawing has $O(n)$ crossings since the graph has $O(n)$ edges and each edge has at most four crossings. The set of lenses in the initial drawing can be identified in $O(n)$ time. Due to the elimination of a single lens, a constant number of other lenses can be affected, which can be computed in constant time. Further, each elimination operation strictly decreases the total number of crossings in the drawing. Consequently, Algorithm~2 performs $O(n)$ elimination operations and can be implemented in $O(n)$ time.

\begin{lemma}\label{lem:0lens}
Let 
$L(e_{\alpha\beta},f_{\alpha\beta})$ be a 0-lens. 
Then operation \textsc{Reroute}$(e_{\alpha\beta},f_{\alpha\beta})$ decreases the total number of crossings and does not create any new crossing.
Further, if any two edges have at most two points in common, then
it does not create any new lens.
\end{lemma}
\begin{proof}
  The operation \textsc{Reroute}$(e_{\alpha\beta},f_{\alpha\beta})$
  modifies only the edge $f$, by rerouting the arc $f_{\alpha\beta}$
  to closely follow $e_{\alpha\beta}$. Since the arc $e_{\alpha\beta}$
  is crossing-free, the edge $f$ loses one of its crossings and no edge
  gains any new crossing. Overall, the total number of crossings
  decreases,
  as claimed.

  Assume that any two edges have at most two points in common before the operation.
  Consider a lens $L(g_{\gamma\delta},h_{\gamma\delta})$ in the drawing after the operation. As no new crossings are created, $\gamma$ and $\delta$ are already common points of $g$ and $h$ before the operation. Since $g$ and $h$ have no other common points by assumption, the lens $L(g_{\gamma\delta},h_{\gamma\delta})$ is already present before the operation. 
\end{proof}

For quasi-0-lenses we define the operation
\textsc{Quasi-0-Reroute}$(e_{\alpha\beta},f_{\alpha\beta})$ as follows; see \figurename~\ref{fig:quasilens}.
Let $h$ be the edge that crosses $e_{\alpha\beta}$ at $\gamma$ and shares an endpoint $s$ with $f$.
%
%
Redraw $f$ such that it closely follows $h$ from $s$ to $\overline{\gamma}$, it does not cross $e$ in $\overline{\gamma}$, and then it closely follows arc $e_{\alpha\beta}$ to $\overline{\beta}$, and further follows its original arc from $\overline{\beta}$ to $t$.
The analogue of Lemma~\ref{lem:0lens} reads as follows.

\begin{figure}[bht]
  \centering
  \subfigure{
	 \label{fig:quasi0lens-1}
	\includegraphics[scale=1]
	{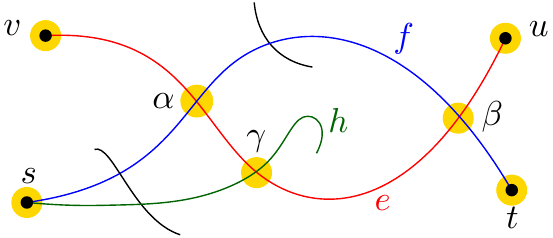}}
  \hfill
  \subfigure{
	 \label{fig:quasi0lens-2}
	\includegraphics[scale=1]
	{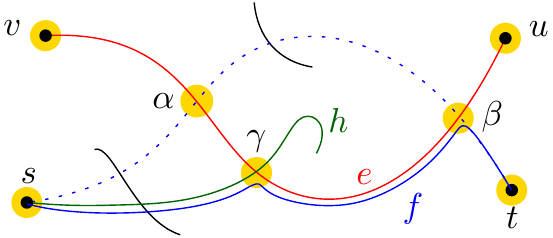}}

    \caption{\textsc{Quasi-0-Reroute}$(e_{\alpha\beta},f_{\alpha\beta})$ applied to a
      quasi-0-lens $L(e_{\alpha\beta},f_{\alpha\beta})$.}
  \label{fig:quasilens}
\end{figure}

\begin{lemma}\label{lem:q0lens}
  Let 
  $L(e_{\alpha\beta},f_{\alpha\beta})$ be a quasi-0-lens, where $h$
  denotes the edge that crosses $e_{\alpha\beta}$ at
  $\gamma$ and shares an endpoint $s$ with $f$. Then
  operation \textsc{Quasi-0-Reroute}$(e_{\alpha\beta},f_{\alpha\beta})$
 decreases the total number of crossings, and does not increase the number of
 crossings between any pair of edges.
Further, if any two edges have at most two points in common, then
it does not create any new lens.
\end{lemma}
\begin{proof}
  The operation \textsc{Quasi-0-Reroute}$(e_{\alpha\beta},f_{\alpha\beta})$ modifies
  only the edge $f$, by rerouting the arc $f_{s\beta}$ to closely
  follow first $h$ from $s$ to $\overline{\gamma}$ and then
  $e_{\alpha\beta}$ to $\overline{\beta}$.
  Let $f'$ denote the new drawing of $f$. Since (1)~$e$ has at least one fewer crossing with $f'$ than with $f$,
  and (2)~every crossing of $f'$ along the arc between $s$ and $\overline{\gamma}$ corresponds to a crossing
  of $f$ along the arc from $s$ to $\overline{\alpha}$, the total number of crossings strictly
  decreases, and for each pair of edges the number of crossings between them does not increase, as claimed.

  Assume that any two edges have at most two points in common before the operation.
  Suppose \textsc{Quasi-0-Reroute}$(e_{\alpha\beta},f_{\alpha\beta})$ creates a new lens.
  This lens must be formed by $f'$ and another edge, say $g$. Then $f'$ and $g$ must have at least two points in common,
  and $g$ must cross $f_{\alpha\beta}$, implying that $f$ and $g$ have at least three points in common before the operation.
  However, by assumption, edges $f$ and $g$ have at most two points in common, which is a contradiction. Consequently, every lens in the resulting drawing corresponds to a lens in the original drawing, where the arc
  $f_{s\alpha}$ is shifted to the arc of $f'$ from $s$ to $\overline{\gamma}$.
\end{proof}

\begin{lemma}\label{lem:exchange}
  Let 
  $L(e_{\alpha\beta},f_{\alpha\beta})$ be a lens either bounded by
  nonadjacent arcs with $x(e_{\alpha\beta})\leq x(f_{\alpha\beta})\leq
  x(e_{\alpha\beta})+2$, or by adjacent arcs with
  $x(e_{\alpha\beta})\leq x(f_{\alpha\beta})\leq
  x(e_{\alpha\beta})+1$. Then the operation \textsc{Swap}$(e_{\alpha\beta},f_{\alpha\beta})$ produces a drawing in which the total number of crossings on each edge does not increase, and the total number of crossings decreases.
\end{lemma}
\begin{proof}
The operation \textsc{Swap}$(e_{\alpha\beta},f_{\alpha\beta})$ modifies only the edges $e$ and $f$, by exchanging arcs $e_{\alpha\beta}$ and $f_{\alpha\beta}$, and eliminating any crossing at the endpoints of these arcs. In particular, the number of crossings on other edges cannot increase. This already implies that the total number of crossings decreases.

Let $e'$ and $f'$ denote the new drawing of $e$ and $f$. If both $\alpha$ and $\beta$ are crossings, then both crossings are eliminated, hence $x(e')=x(e)-2+(x(f_{\alpha\beta})-x(e_{\alpha\beta}))\leq x(e)$ and
$x(f')=x(f)-2+(x(e_{\alpha\beta})-x(f_{\alpha\beta}))\leq x(f)-2$. If $\alpha$ or $\beta$ is a vertex of $G$,  then only one crossing is eliminated, hence $x(e')=x(e)-1+(x(f_{\alpha\beta})-x(e_{\alpha\beta}))\leq x(e)$ and
$x(f')=x(f)-1+(x(e_{\alpha\beta})-x(f_{\alpha\beta}))\leq x(f)-1$, as required.
\end{proof}

\begin{lemma}\label{lem:6}
Let $D$ be a 4-plane drawing of a graph, and let $L(e_{\alpha\beta},f_{\alpha\beta})$ be a lens with $x(e_{\alpha\beta})\leq x(f_{\alpha\beta})$.
\begin{enumerate}
\item If $x(f_{\alpha\beta})-x(e_{\alpha\beta})\geq 2$, then $L(e_{\alpha\beta},f_{\alpha\beta})$ is either a 0-lens or $x(e_{\alpha\beta})=1$ and $x(f_{\alpha\beta})=3$.
\item If $x(e_{\alpha\beta})=1$ and $x(f_{\alpha\beta})=3$, then $e_{\alpha\beta}$ and $f_{\alpha\beta}$ are adjacent arcs.
\end{enumerate}
\end{lemma}
\begin{proof}
As $D$ is 4-plane, we have $x(e)\leq 4$ and $x(f)\leq 4$.
Assume first that both $\alpha$ and $\beta$ are crossings, and so $x(f_{\alpha\beta})\leq x(f)-2\leq 2$. Combined with  $x(f_{\alpha\beta})-x(e_{\alpha\beta})\geq 2$, this implies $x(e_{\alpha\beta})=0$,
hence $L(e_{\alpha\beta},f_{\alpha\beta})$ is a 0-lens. Assume next
that $\alpha$ or $\beta$ is a vertex in $G$. Then
$x(f_{\alpha\beta})\leq x(f)-1\leq 3$. With
$x(f_{\alpha\beta})-x(e_{\alpha\beta})\geq 2$, this implies
$x(e_{\alpha\beta})=0$, or $x(e_{\alpha\beta})=1$ and
$x(f_{\alpha\beta})=3$. 
\end{proof}

\begin{description}
\item[Algorithm~2.]
\item[Input.] Let $D_0$ be a 4-plane drawing of a graph $G=(V,E)$.
\item[Phase 1.] While there is a lens $L(e_{\alpha\beta},f_{\alpha\beta})$ that is not a 1-3-lens, do:\\
If it is a 0-lens, then \textsc{Reroute}$(e_{\alpha\beta},f_{\alpha\beta})$,
else \textsc{Swap}$(e_{\alpha\beta},f_{\alpha\beta})$.
\item[Phase~2.] Let $\mathcal{L}$ be the set of 1-3-lenses.
For every $L(e_{\alpha\beta},f_{\alpha\beta})\in \mathcal{L}$, if
neither $e_{\alpha\beta}$ nor $f_{\alpha\beta}$ has been modified in
previous iterations of Phase~2 (regardless of whether
$x(e_{\alpha\beta})$ or $x(f_{\alpha\beta})$ has changed), apply
\textsc{Reroute}$(e_{\alpha\beta},f_{\alpha\beta})$.
\item[Phase~3.] While there is a 0-lens
  $L(e_{\alpha\beta},f_{\alpha\beta})$, do:
  \textsc{Reroute}$(e_{\alpha\beta},f_{\alpha\beta})$.\\
  While there is a quasi-0-lens $L(e_{\alpha\beta},f_{\alpha\beta})$,
  do: \textsc{Quasi-0-Reroute}$(e_{\alpha\beta},f_{\alpha\beta})$.
\end{description}

For $i\in \{1,2,3\}$, let $D_i$ denote the drawing obtained at the end of Phase~$i$.
We analyse the three phases separately.

\begin{lemma}\label{lem:pre}
Phase~1 terminates, and $D_1$ is a 4-plane drawing in which every lens is a 1-3-lens,
and any two edges have at most two points in common.
\end{lemma}
\begin{proof}
  By Lemma~\ref{lem:0lens} and Observation~\ref{obs:exchange}, each
  iteration of the while loop reduces the total number of
  crossings. Since $D_0$ has at most $\frac12\cdot 4m\in O(n)$
  crossings, the while loop terminates after $O(n)$ iterations. By
  Lemma~\ref{lem:6} all lenses satisfy the conditions of
  Lemma~\ref{lem:exchange}, except for 0-lenses and lenses
  $L(e_{\alpha\beta},f_{\alpha\beta})$ with $x(e_{\alpha\beta})=1$ and
  $x(f_{\alpha\beta})=3$. Each lens of the latter type is either a
  1-3-lens, which remains untouched, or $x(e)< 4$ and the lens is
  eliminated by a swap operation. In this case, though the number of
  crossings on the edge $e$ increases, it does not exceed four and
  the total number of crossings in the drawing strictly decreases.
  In all other cases we can apply either Lemma~\ref{lem:0lens} or
  Lemma~\ref{lem:exchange} to conclude that each iteration maintains a
  4-plane drawing. By the end condition of the while loop, all lenses
  other than 1-3-lenses are eliminated.

  To prove the final statement, suppose to the contrary, two edges $e$
  and $f$ in $D_1$ have three or more points in common.
  By Lemma~\ref{lemma:existence_of_lens}, there exist arcs
  $e_{\alpha\beta}\subset e$ and $f_{\alpha\beta}\subset f$
  such that $L(e_{\alpha\beta},f_{\alpha\beta})$ is a lens, which is necessarily a 1-3-lens.
  We may assume without loss of generality that $x(e)=4$, $x(e_{\alpha\beta})=1$, and $x(f_{\alpha\beta})=3$.
  Denote by $\gamma$ a common point of $e$ and $f$ other than $\alpha$ and $\beta$.
  Since $D_1$ is a 4-plane drawing and $x(f_{\alpha\beta})=3$, we may assume that $\alpha$ is common endpoint of $e$ and $f$,
  furthermore $\gamma$ is a crossing in the interior of $f_{\alpha\beta}$.
  Since $e_{\alpha\beta}$ and $f_{\alpha\beta}$ form a lens, the arc $e_{\alpha\beta}$ cannot pass through $\gamma$.
  Hence $\gamma$ is a crossing between $f_{\alpha\beta}$ and $e\setminus e_{\alpha,\beta}$.
  By Lemma~\ref{lemma:existence_of_lens}, $e_{\beta\gamma}$ and $f_{\beta\gamma}$ form a lens,
  which is necessarily a 1-3-lens. However, $x(e_{\beta\gamma})\leq 2$ and $x(f_{\beta\gamma})\leq 2$,
  which is a contradiction.
\end{proof}

For the analysis of Phases~2 and 3, we introduce some notation. Let $N$ denote the planarization of $D_1$. Note that $N$ is a simple graph, since a double edge would correspond to a lens whose arcs are crossing-free (i.e., a 0-lens). Phases 2 and 3 apply only \textsc{Reroute} and \textsc{Quasi-0-Reroute} operations. Hence the resulting drawings satisfy invariants \ref{I1}--\ref{I3}.
For a node $\alpha$ of $N$, we denote by $\overline{\alpha}$ a small neighborhood of $\alpha$.
Recall that the length $\ell(a)$ of an arc $a$ along an edge of $G$ is the combinatorial length of the path in $N$ that the arc closely follows.

\begin{lemma}\label{lem:lengths}
  $D_2$ has the following properties:
  (i) the length of every edge is at most five;
  (ii) at most two edges of $G$ pass along every segment of $N$;
  (iii)~through every node $\nu$ of $N$, at most two rerouted edges of $G$ pass through $\nu$; and
  (iv) at each node $\alpha$ of $N$, an edge passing through
    $\overline{\alpha}$ crosses at most two edges in $\overline{\alpha}$;
  (v)~any two edges have at most two points in common.
\end{lemma}
\begin{proof}
\textbf{(i)} By Lemma~\ref{lem:pre}, the drawing $D_1$ is a 4-plane drawing. Therefore, every edge in $D_1$ passes through at most 4 crossings, hence its length is at most 5. Each \textsc{Reroute} operation in Phase~2 replaces an edge of length 5 with an edge of length 3 (cf.~\figurename~\ref{fig:lemma_for_4_img1}). Property (i) follows.

\noindent\textbf{(ii)}
Each \textsc{Reroute}$(e_{\alpha\beta},f_{\alpha\beta})$ operation in Phase~2 reroutes the longer arc along the shorter arc of a 1-3-lens in $\mathcal{L}$. Let $\mathcal{A}$ be the set paths of length 2 in $N$ that correspond to shorter arcs $e_{\alpha\beta}$ in some 1-3-lens $L(e_{\alpha\beta},f_{\alpha\beta})\in \mathcal{L}$.
By the definition of 1-3-lenses, $\ell(e)=5$ and $e_{\alpha\beta}$ consists of the first two segments of $N$ along $e$. Thus every segment $\gamma\delta$ of $N$ is contained in at most one path in $\mathcal{A}$.
Consequently, at most one new edge can pass along $\gamma\delta$ due to reroute operations.
%

\noindent\textbf{(iii)} Let $\gamma$ be a node in $N$ that corresponds to a crossing in the drawing $D_1$.
Then $\gamma$ is incident to at most two paths in $\mathcal{A}$ (at most one along each of the two edges that cross at $\gamma$). Hence at most two rerouted edges can pass through $\gamma$.

\noindent\textbf{(iv)}
Let $\gamma$ be a node of $N$, and let $e$ be an edge that passes
through $\overline{\gamma}$ in $D_2$.  By property~(ii), at most $4$
edges pass through $\overline{\gamma}$. If at most $3$ edges pass
through $\overline{\gamma}$, then it is clear that $e$ crosses at most
two edges in $\overline{\gamma}$. Suppose that four edges pass through
$\overline{\gamma}$. Then the four segments of $N$ incident to
$\gamma$ are each contained in the shorter arc of some 1-3-lens in
$D_1$. Consequently, $\gamma$ is the middle vertex of two distinct arcs
in $\mathcal{A}$. In the drawing $D_2$ (after \textsc{Reroute} operations),
two edges run in parallel in each of these shorter arcs.
Hence each edge that passes through $\overline{\gamma}$
crosses at most two other edges in $\overline{\gamma}$, as claimed.

\noindent\textbf{(v)} Suppose $f_1$ and $f_2$ have three points in common in $D_2$. By Lemma~\ref{lem:pre}, we may assume that $f_1$ has been rerouted in Phase~2, and $f_1$ follows a path $(v,\gamma,\delta,u)$ in $N$ and $(v,\gamma,\delta)\in \mathcal{A}$. Since $f_1$ and $f_2$ have at most one common endpoint, they cross in both $\overline{\gamma}$ and $\overline{\delta}$.
After the rerouting operation, $f_1$ does not cross any edge of $D_1$ in $\overline{\delta}$, which implies that $f_2$ has also been rerouted in Phase~2.
Since both $f_1$ and $f_2$ have length three and pass through $\overline{\gamma}$ and $\overline{\delta}$, and $N$ is a simple graph, both $f_1$ and $f_2$ pass along segment $\gamma\delta$,
which contradicts the fact that at most one new edge can pass along $\gamma\delta$ (see the proof of~(ii) above).
\end{proof}

\begin{corollary}\label{cor:8plane}
$D_2$ is an 8-plane drawing of $G$.
\end{corollary}
\begin{proof}
Every edge of $G$ passes through the small neighborhood of at most four nodes of $N$ by Lemma~\ref{lem:lengths}(i). In each such neighborhood, it crosses at most two other edges by Lemma~\ref{lem:lengths}(iv), and it has at most one crossing with each by \ref{I3}.
Overall, every edge has at most eight crossings in $D_2$.
\end{proof}

Unfortunately, Phase~2 may create new lenses, but only of very
specific types. 
We analyze these types and argue that all remaining lenses are removed. 

\begin{lemma}\label{lem:phase3}
Phase~3 terminates 
with an 8-plane simple topological drawing $D_3$.
\end{lemma}
\begin{proof}
  The while loops in Phase~3 terminate, as each iteration decreases
  the number of crossings by
  Lemmas~\ref{lem:0lens} and~\ref{lem:q0lens}. The drawing $D_2$ at the
  beginning of Phase~3 is 8-plane by Corollary~\ref{cor:8plane}, and
  remains 8-plane and no new lens is created by
  Lemmas~\ref{lem:0lens} and~\ref{lem:q0lens}.
  It remains to show that Phase~3 
  eliminates all lenses
  of $D_2$.

  Every lens in $D_1$ is a 1-3-lens by Lemma~\ref{lem:pre}, and they
  are all in $\mathcal{L}$. Phase~2 modifies an arc in every lens in
  $\mathcal{L}$. Thus the lenses of $D_1$ are no longer present in
  $D_2$. (The two edges that form a lens $L\in\mathcal{L}$ may still
  form a lens $L'$ in $D_2$, but technically this is a new lens, that
  is, $L\ne L'$, which is created in Phase~2 and will be discussed
  next.)

  We classify the new lenses created in Phase~2.
  Assume that edges $e$ and $f=uw$ form a lens in $D_2$. Without loss
  of generality, the edge $f$ was modified in Phase~2. Each iteration
  in Phase~2 applies a reroute operation on a 1-3-lens, which
  decreases the length of an edge from 5 to 3. Therefore Phase~2
  modifies every edge at most once. The drawing of edge $f$ in $D_2$
  was produced by a \textsc{Reroute}$(g_{u\beta},f_{u\beta})$
  operation, for some edge $g$, where $u$ is a common endpoint of $f$
  and $g$. The resulting drawing of $f$ in $D_2$ closely follows a
  path $(u,\alpha,\beta)$ in $N$ and then the original arc
  (in $D_1$) from $\overline{\beta}$ to $w$.
  After operation \textsc{Reroute}$(g_{u\beta},f_{u\beta})$,
  edges $f$ and $g$ do not cross each other.

  Suppose first that $f$ crosses $e$ in $\overline{\beta}$. Then $e$
  was redrawn in Phase~2 to closely follow $f$ from $w$ to
  $\overline{\beta}$ and beyond; as in
  \figurename~\ref{fig:interlock1}. However, in this case, $e$ and $f$
  have a common endpoint at $w$. No other edges follow segment $\beta w$ in $N$
  by Lemma~\ref{lem:lengths}(ii), hence $e$ and $f$ form a 0-lens.
  All such 0-lenses are eliminated in Phase~3, without creating any new lenses
  (cf.~Lemma~\ref{lem:0lens}). Therefore, we may assume that $f$ does
  not cross any edge in $\overline{\beta}$.

\begin{figure}[ht]
  \centering
  \includegraphics[scale=1]
  {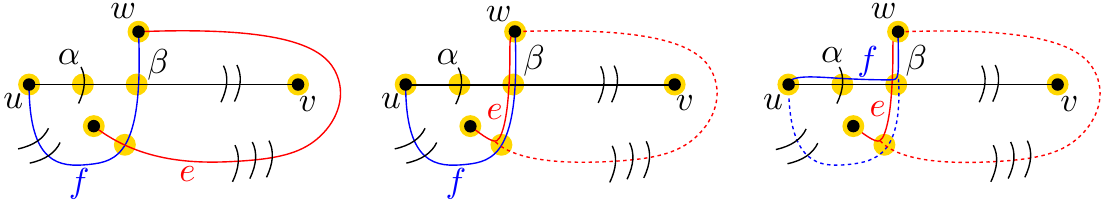}
  \caption{New 0-lens formed by $e$ and $f$ crossing in $\overline{\beta}$.} 
  \label{fig:interlock1}
\end{figure}

By Lemma~\ref{lem:lengths}(iv), the edge $f$ crosses at most two other
edges in $\overline{\alpha}$. If it crosses exactly one other edge,
and $f$ forms a lens $L$ with that edge, then this crossing in
$\overline{\alpha}$ is the only crossing of $f$ in $D_2$ and, thus,
$L$ is a 0-lens. Otherwise, $f$ crosses two edges, denote them by $e$
(for which we know that it crosses $f$) and $h$; one of them was
redrawn in a \textsc{Reroute} operation in Phase~2 to closely follow
the other, which passes through $\overline{\alpha}$; see
\figurename~\ref{fig:interlock4} and~\ref{fig:interlock3}. Therefore, $e$
and $h$ are adjacent, and they do not cross at the end of that
operation. Thus, they do not cross in $D_2$, either; otherwise, three
rerouted edges would pass through $\overline{\alpha}$, contradicting
Lemma~\ref{lem:lengths}(iii). As no new crossing is introduced in
Phase~3, the edges $e$ and $h$ do not cross anytime during (and after)
Phase~3, either.


\begin{figure}[thb]
  \centering
  \includegraphics[scale=1]
  {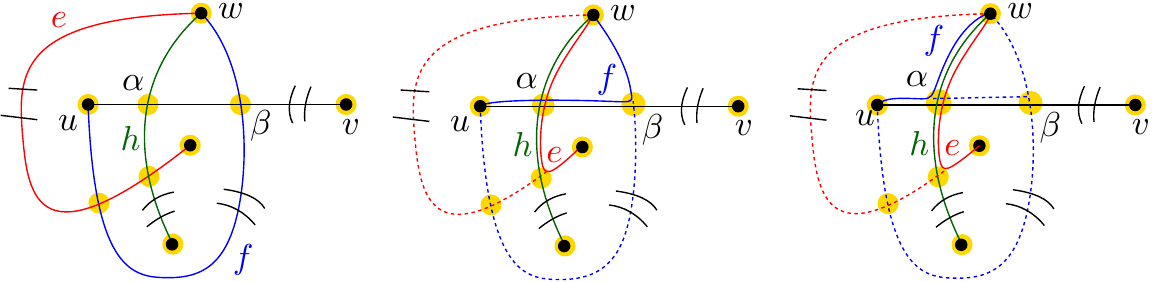}
  \caption{$f$ crosses two edges at $\overline{\alpha}$ and forms two 0-lenses.}
  \label{fig:interlock4}
\end{figure}

If the common endpoint of $e$ and $h$ is $u$ or $w$ (see
\figurename~\ref{fig:interlock4}), then both $e$ and $h$ form a lens
with $f$: One of these lenses is a 0-lens, and when this lens is
eliminated, the other lens either disappears, or it becomes a 0-lens
as well. Hence Phase~3 eliminates both crossings.

If $e$ and $h$ share distinct endpoints with $f$, without loss of
generality $e$ and $f$ are adjacent at $u$ and $h$ and $f$ are
adjacent at $w$. As $e$ and $h$ do not cross, the crossing $e\cap f$
is closer to $u$ and the crossing $h\cap f$ is closer to $w$ along
$f$. Hence, $e$ and $h$ each form a 0-lens with $f$, both of which are
eliminated in Phase~3.

\begin{figure}[ht]
  \centering
  \includegraphics[scale=1]
  {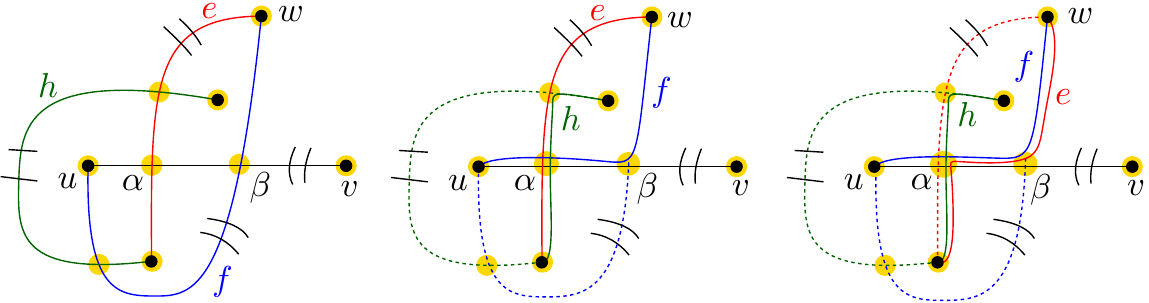}
  \caption{$f$ crosses two edges at $\overline{\alpha}$ and forms a quasi-0-lens.}
  \label{fig:interlock3}
\end{figure}

It remains to consider the case that only $e$ is adjacent to $f$
(while $h$ is not). Assume first that $e$ and $f$ are adjacent at $w$
(see \figurename~\ref{fig:interlock3}). If the crossing $e\cap f$ is
closer to $w$ along $f$ than the crossing $h\cap f$, then the lens
formed by $e$ and $f$ is a 0-lens; else it forms a quasi-0-lens. In
any case, the lens is eliminated in Phase~3. The same argument works
in case that $e$ and $f$ are adjacent at $u$.
\end{proof}


\clearpage\appendix

\section{Proof of Lemma~\ref{lemma:existence_of_lens}}

\lemmaone*\begin{proof}
Suppose, to the contrary, that there are no two intersection points that determine a lens. Then the arcs between any two intersection points cross each other. Consider two intersection points, $\alpha$ and $\beta$, for which the arcs $e_{\alpha\beta}\cup f_{\alpha\beta}$ have the minimum number of crossings. Let $\gamma$ be one of these crossings. Then the arcs $e_{\alpha\gamma}\subset e_{\alpha\beta}$ and $f_{\alpha\gamma}\subset f_{\alpha\beta}$ have fewer crossings, contradicting the minimality assumption in the choice of $\alpha$ and $\beta$.
\end{proof}

\section{Proof of Observation~\ref{obs:exchange}}

\obsone*\begin{proof}
  The swap operation modifies only two edges, and step~(3) removes any
  self-in\-ter\-sec\-tions introduced in previous steps. Consequently,
  the output is a topological drawing.  The operation eliminates
  crossings at both $\alpha$ and $\beta$. At least one of $\alpha$ and
  $\beta$ is a crossing (one of them may be a common vertex), so at
  least one crossing is eliminated. On the other hand, the operation
  does not create any new crossings, hence the total number of
  crossings drops by at least one.
\end{proof}

\section{Proof of Observation~\ref{obs:reroute}}

\obstwo*\begin{proof}
  Only one edge is modified and Step~(3) removes any
  self-in\-ter\-sec\-tions introduced in previous steps. Thus the
  output is a topological drawing.
\end{proof}

\section{Proof of Observation~\ref{obs:length}}
\obsthree*\begin{proof}
By construction, \textsc{Reroute} operations maintain invariants \ref{I1}--\ref{I3}.
Initially, $D_0$ is a $k$-plane drawing, so every edge corresponds to a path of length at most $k+1$ in the planarization $N$ of $D_0$. Each iteration of the while loop replaces an arc with another arc of the same or smaller length. The claim follows.
\end{proof}

\section{Remarks about Algorithm~1}


\noindent\textbf{(1)} Constant-factor improvements to the bound
$f(k)\leq \frac23 \sqrt{58} \cdot k^{3/2}\cdot 3^k$ are not difficult to
obtain, at the expense of making the algorithm and analysis more complicated. We mention two possible improvements.

\noindent (a)
In the bound $n_\gamma \leq 4\cdot 3^{k-2}$, we assume that \emph{all} endpoints of the edges that pass through $\overline{\gamma}$ are at distance $k$ from $\gamma$ in $N$, which is impossible. However, one could have 2 vertices at distance 1 from $\gamma$, and $2\cdot 3^{k-2}$ vertices at distance $k$. A more careful analysis might save up to a factor of 2.

\noindent (b) We could modify Algorithm~1 so that we apply operation
\textsc{Swap}$(e_{\alpha\beta},f_{\alpha\beta})$ if
$\ell(e_{\alpha\beta})=\ell(f_{\alpha\beta})$, and operation
\textsc{Reroute}$(e_{\alpha\beta},f_{\alpha\beta})$ if
$\ell(e_{\alpha\beta})=\ell(f_{\alpha\beta})$ and
$x(e_{\alpha\beta})\leq x(f_{\alpha\beta})$. Note that under
\textsc{Exchange} operations, the number of edges passing though a
neighborhood $\overline{\gamma}$ does not increase. Using this modified
algorithm, the length of every edge that has been rerouted is at most
$k$. So $\gamma$ is at distance at most $k-1$ from the endpoints of
edges that pass through $\overline{\gamma}$ (with the possible
exception of the two edges that cross at $\gamma$ in the original
drawing). This would imply $n_\gamma \leq 4\cdot 3^{k-2}$, improving the bound on $f(k)$ by a factor of 3.

\noindent\textbf{(2)} Algorithm~1 incrementally modifies the edges of a drawing to eliminate lenses. An alternative algorithm, which follows a global redrawing strategy, would yield essentially the same bound on $f(k)$, and ensure that every edge in the resulting simple topological drawing closely follows a shortest path in $N$. Specifically, we could label the segments of $N$ by $s_1,\ldots, s_{t}$, where $t\in O(kn)$ is the number of segments in $N$, and assign a weight $w(s_i)=2^i+2^{t+1}$ to every segment. For each edge $e=(u,v)\in E(G)$, consider the network $N_{uv}-(V\setminus \{u,v\})$, and then draw $e$ so that it closely follows the weighted path in $N$ from $u$ to $v$. The weights guarantee that this is also an unweighted path in $N_{uv}$. Furthermore, it is not difficult to show that if $L(e_{\alpha\beta},f_{\alpha\beta})$ is a lens, then both $e_{\alpha\beta}$ and $f_{\alpha\beta}$ follow the same path in $N$; and all such lenses can successively be eliminated by \textsc{Exchange} operations.

\noindent\textbf{(3}) We do not know whether the upper bound $f(k)\in \exp(O(k))$ can be improved to a bound polynomial in $k$. However, Algorithm~1 (as well as the global ``shortest path'' approach mentioned above) does not yield a sub-exponential bound, as there exist $k$-plane drawings for which these algorithms return simple topological graphs whose local crossing numbers are exponential in $k$ (private communication with Bal\'azs Keszegh).

\end{document}